\documentclass[11pt]{article}

\usepackage[english]{babel}
\usepackage[margin=1in]{geometry}
\usepackage{amsmath,amssymb,amsthm}
\usepackage[utf8]{inputenc}
\usepackage{tikz}
\usepackage{xspace}
\usepackage{subfig}
\usepackage{paralist}
\usepackage{hyperref}
\usepackage{times}
\newcommand{\1}[1]{{\normalfont \ensuremath{#1_{1}}}} %
\newcommand{\G}{\ensuremath{\mathcal{G}}}

\newcommand{\mpd}{\textsc{Mutual Planar Duality}\xspace}
\newcommand{\threepart}{\textsc{3-Partition}\xspace}

\DeclareMathOperator{\skel}{skel}

\DeclareMathOperator{\corr}{corr}

\newcommand{\eps}{\varepsilon}

\newtheorem{theorem}{Theorem}
\newtheorem{lemma}{Lemma}
\newtheorem{corollary}{Corollary}

\begin{document}

%\setpagewiselinenumbers
%\linenumbers

\title{Testing Mutual Duality of Planar Graphs}

\author{Patrizio Angelini$^1$, Thomas Bl\"{a}sius$^2$, and Ignaz Rutter$^2$
\\\\
\emph{$^1$Universit\`a  Roma Tre, Italy}\\
{\small angelini@dia.uniroma3.it}\\
\emph{$^2$Faculty of Informatics, Karlsruhe Institute of Technology (KIT),
Germany}\\
{\small \{thomas.blaesius,rutter\}@kit.edu}}

\date{}

\maketitle

\begin{abstract} 
  We introduce and study the problem \mpd, which asks for two planar
  graphs~$G_1$ and~$G_2$ whether~$G_1$ can be embedded such that its
  dual is isomorphic to~$G_2$.  Our algorithmic main result is an
  NP-completeness proof for the general case and a linear-time
  algorithm for biconnected graphs.

  To shed light onto the combinatorial structure of the duals of a
  planar graph, we consider the \emph{common dual relation}~$\sim$,
  where $G_1 \sim G_2$ if and only if they have a common dual.
  While~$\sim$ is generally not transitive, we show that the
  restriction to biconnected graphs is an equivalence relation.  In
  this case, being dual to each other carries over to the equivalence
  classes, i.e., two graphs are dual to each other if and only if any
  two elements of their respective equivalence classes are dual to
  each other.
  To achieve the efficient testing algorithm for \mpd on biconnected
  graphs, we devise a succinct representation of the equivalence class
  of a biconnected planar graph.  It is similar to SPQR-trees and
  represents exactly the graphs that are contained in the equivalence
  class.  The testing algorithm then works by testing in linear time
  whether two such representations are isomorphic.

  We note that a special case of \mpd is testing whether a graph~$G$
  is self-dual.  Our algorithm handles the case where~$G$ is
  biconnected and our NP-hardness proof extends to testing
  self-duality of general planar graphs and also to testing map
  self-duality, where a graph~$G$ is map self-dual if it admits a
  planar embedding~$\mathcal G$ such that~$G^\star$ is isomorphic
  to~$G$, and additionally the embedding induced by~$\mathcal G$
  on~$G^\star$ is~$\mathcal G$.
\end{abstract}

\setcounter{page}{0}
\thispagestyle{empty}
\newpage

\section{Introduction}
\label{sec:introduction}

Let~$G$ be a planar graph with planar embedding~$\G$ and let~$F$
be the set of faces of~$\G$.  The \emph{dual} of~$G$ with respect
to~$\G$ is the graph~$G^\star = (F,E^\star)$, where~$E^\star = \{
e^\star \mid e \in E\}$ and~$e^\star$ denotes the edge connecting the
two faces incident to~$e$ in~$\G$.  Thus~$G^\star$
models the adjacencies of faces of~$G$ in the embedding~$\G$.  Note
that the circular order of edges around faces in~$\G$ naturally induces a
planar embedding~$\G^\star$ on~$G^\star$ and that the dual of~$G^\star$
with respect to~$\G^\star$ is~$G$.

We consider the following problem, that we call \mpd.  Given two planar
graphs~$G_1$ and~$G_2$, is it possible to find an embedding~$\G_1$
of~$G_1$ such that the dual $G_1^\star$ of~$G_1$ with respect
to~$\G_1$ is isomorphic to~$G_2$?  If $G_1$ is triconnected it has a
fixed planar embedding~\cite{w-cgcg-32} and thus the problem \mpd
reduces to testing graph isomorphism for planar graphs, which can be
done in linear time due to Hopcroft and Wong~\cite{hw-ltaipg-74}.
Observe that biconnectivity and triconnectivity of a planar graph is
invariant under dualization~\cite{t-cm-66}.  For non-triconnected
planar graphs \mpd is more complicated since changing the embedding of
$G_1$ influences the structure of its dual graph.  In fact, we show
that \mpd is NP-complete in general. 

On the other hand, for biconnected planar graphs we provide a
linear-time algorithm solving \mpd that is based on the definition of
a new data structure that we call \emph{dual SPQR-tree} in analogy to
the SPQR-tree~\cite{dt-omtc-96,dt-opt-96}. As SPQR-trees allow to
succinctly represent and efficiently handle all the planar embeddings
of a biconnected planar graph, the dual SPQR-trees, together with a
newly-defined set of operations, allow to succinctly represents and
efficiently handle all the dual graphs of a biconnected planar graph.
This data structure has an interesting implication on the structure of
the dual graphs of a biconnected planar graph.  Namely, consider the
\emph{common dual relation} $\sim$, where $G_1 \sim G_2$ if and only
if they have a common dual graph.  We show that, $\sim$ is not
transitive on the set of connected planar graphs.  However, it follows
from the dual SPQR-tree that $\sim$ is an equivalence relation on the
set of biconnected planar graphs.  In particular, the graphs
represented by a dual SPQR-tree form an equivalence class.  Thus,
testing \mpd reduces to testing whether two dual SPQR-trees represent
the same equivalence class.

We believe that this new data structure can be successfully used to
efficiently solve other related problems. In fact, in many applications it
is desirable to find an embedding of a given planar graph that optimizes
certain criteria, and such criteria can often be naturally described in
terms of the dual graph with respect to the chosen embedding. For example,
Bienstock and Monma~\cite{bm-cepgmcdm-90}, and Angelini et
al.~\cite{abp-mepg-11} seek for an embedding of a given planar graph
minimizing the largest distance of internal faces to the external face. In
terms of the dual graph this corresponds to minimizing the largest distance
of a vertex to all other vertices. Hence, for problems of this kind it
might be useful to work directly with a representation of all dual graphs,
that can be given by the dual SPQR-trees, instead of taking the detour over
a representation of all planar embeddings, that is given by the original
SPQR-trees.

We finally remark that the \mpd problem we introduce in this paper is a
generalization of the self-duality of planar graphs~\cite{sc-csdg-92}. A
graph $G$ is \emph{graph self-dual} if it admits an embedding such that its
dual $G^\star$ is isomorphic to $G$. We call the corresponding decision
problem {\sc Graph Self-Duality}. A stronger form of self-duality can
be defined as follows. A graph $G$ is \emph{map self-dual}~\cite{ss-sdg-96}
if and only if $G$ has an embedding $\G$ such that there exists an
isomorphism from $G$ to its dual graph $G^\star$ that preserves the
embedding $\G$. The corresponding decision problem is called {\sc
Map Self-Duality}. Note that, since triconnected planar graphs have a
unique planar embedding, {\sc Graph Self-Duality} and {\sc Map
Self-Duality} are equivalent for this class of graphs.
Servatius and Servatius~\cite{ss-sdg-96} show the
existence of
biconnected planar graphs that are graph self-dual but not map
self-dual. Servatius and Christopher~\cite{sc-csdg-92} show how to
construct self-dual graphs from given planar graphs.  Archdeacon and
Richter~\cite{ar-ccsdsp-92} give a set of constructions for
triconnected self-dual graphs and show that every such graph can be
constructed in this way.  To the best of our knowledge the
computational complexity of testing {\sc Map} or {\sc Graph Self-Duality}
is open. Since {\sc Graph Self-Duality} is a special case of \mpd (simply
set $G_1 = G_2 = G$), our algorithm solving \mpd in linear time can be used
to solve {\sc Graph Self-Duality} in linear time when $G$ is biconnected. 
Moreover, our construction showing NP-hardness of \mpd for general
instances extends to {\sc Map} and {\sc Graph Self-Duality}.

% \paragraph{Related Work.}

% An important relation between a graph $G$ and its dual $G^\star$ is
% that there is a one-to-one correspondence between their edges such
% that every cut in $G$ correspond to a cycle in $G^\star$ and vice
% versa.  In fact, this relation can be used to define dual graphs
% without using a planar embedding and Whitney~\cite{w-nspg-32} showed
% that these two definitions are equivalent, which in particular shows
% that a graph is planar if and only if it has a dual graph.  Thus,
% having a dual graph is an important property of planar graphs and dual
% graphs have received much attention in the past century.  Here we only
% cite few results related to our considerations.

\paragraph{Outline.}

In Section~\ref{sec:complexity} we show that \mpd is NP-complete, even
if both input graphs are required to be simple.  With a similar
construction we can show that {\sc Map Self-Duality} and {\sc Graph
  Self-Duality} are NP-complete in general.  To solve \mpd efficiently
for biconnected graphs, we first describe decomposition trees as a
generalization of SPQR-trees in Section~\ref{sec:spqr-tree}.  In
Section~\ref{sec:succ-repr-all-duals} we describe the dual SPQR-tree
and show that it succinctly represents all dual graphs of a
biconnected planar graph.  We consider the common dual relation in
Section~\ref{sec:equivalence-relation} and give a counter example
showing that $\sim$ is not transitive on the set of connected planar
graphs.  On the other hand, we show that it follows from the dual
SPQR-tree that $\sim$ is an equivalence relation on the set of
biconnected planar graphs.  This implies that solving \mpd is
equivalent to testing whether two dual SPQR-trees represent the same
equivalence class.  In Section~\ref{sec:solv-mpd-biconn} we show that
this can be further reduced to testing graph isomorphism of two planar
graphs, which leads to a linear-time algorithm for \mpd, including
{\sc Graph Self-Duality} as a special case.

\section{Complexity}
\label{sec:complexity}

In this section we first show that \mpd is NP-complete by a reduction
from \threepart.  Then we show that the resulting instances of \mpd
can be further reduced to equivalent instances of {\sc Map} and {\sc
  Graph Self-Duality}.  An instance~$(A,B)$ of \threepart consists of
a positive integer~$B$ and a set~$A=\{a_1,\dots,a_{3m}\}$ of~$3m$
integers with~$B/4 < a_i < B/2$ for $i=1,\dots,3m$.  The question is
whether~$A$ admits a partition~$\cal A$ into a set of triplets such
that for each triplet $X \in \cal A$ we have~$\sum_{x \in X} x = B$.
The problem \threepart is strongly NP-hard~\cite{gj-cigtnpc-79}, i.e.,
it remains NP-hard even if~$B$ is bounded by a polynomial in~$m$.

\begin{theorem}
  \label{thm:mpd-np-compl}
  \mpd is NP-complete, even if both graphs are simple.
\end{theorem}
\begin{proof}
  Clearly, \mpd is in~NP as we can guess an embedding for graph~$\1G$
  and then check in polynomial time whether the corresponding dual is
  isomorphic to~$\2G$.

  To show NP-hardness we give a reduction from \threepart.  We first
  give a construction containing loops, afterwards we show how to get
  rid of them.  Let~$(A,B)$ be an instance of \threepart with~$|A| =
  3m$.  The graph~$G_1$ contains a wheel of size $m$, i.e., a cycle
  $v_1, \dots, v_m$ together with a center $u$ connected to each
  $v_i$.  Additionally, for each element~$a_i \in A$ we create a star
  $T_i$ with~$a_i -1$ leaves and connect its center to~$u$; see
  Figure~\ref{fig:reduction}(a).  The graph~$\2G$ is a wheel of
  size~$m$ along with~$B$ loops at every vertex except for the center;
  see Figure~\ref{fig:reduction}(b).  We now claim that~$\1G$
  and~$\2G$ form a \textsc{Yes}-instance of \mpd if and only
  if~$(A,B)$ is a \textsc{Yes}-instance of \threepart.

  \begin{figure}[tb]
    \centering
    % \subfloat[]{\includegraphics[page=1]{reduction}\label{fig:reduction-g1}}\hfil
    % \subfloat[]{\includegraphics[page=2]{reduction}\label{fig:reduction-g2}}\hfil
    % \subfloat[]{\includegraphics[page=3]{reduction}\label{fig:reduction-face}}
    \includegraphics[page=4]{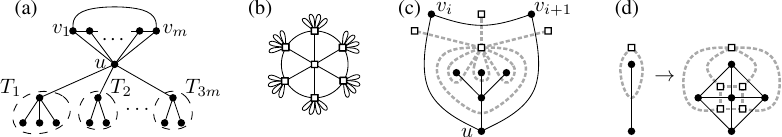}
    \caption{The graphs~$\1G$ (a) and~$\2G$ (b) of the reduction from
      \threepart.  (c) Embedding a tree~$T_i$ inside a face~$f$
      creates~$a_i$ loops at the corresponding dual vertex.  (d)
      Bridges and corresponding loops can be replaced by small
      graphs.}
    \label{fig:reduction}
  \end{figure}

  Suppose that there exists a partition $\cal A$ of $A$.  The
  embedding of the wheel in $G_1$ is fixed and it has exactly $m$
  faces incident to the wheels center $u$.  The remaining degree of
  freedom is to decide the embedding of the trees~$T_i$ into these $m$
  faces.  For each triplet $X = \{a_i,a_j, a_k\} \in \cal A$ we pick a
  distinct such face and embed~$T_i,T_j$ and~$T_k$ into it.  Call the
  resulting embedding~$\1\G$ and consider the dual~$G_1^\star$ with
  respect to~$\1\G$.  The wheel of~$\1G$ determines a wheel of
  size~$m$ in~$G_1^\star$.  Consider a tree~$T_i$ that is embedded in
  a face~$f$.  Since~$T_i$ contains~$a_i$ bridges, which are all
  embedded in~$f$, the corresponding vertex of~$G_1^\star$ has $a_i$
  loops; see Figure~\ref{fig:reduction}(c).  Due to the construction
  each face contains exactly three trees with a total of~$B$ bridges.
  Thus $G_1^\star$ is isomorphic to~$\2G$.

  Conversely, assume that we have an embedding~$\1\G$ such that the
  dual~$G_1^\star$ of~$\1G$ with respect to~$\1\G$ is isomorphic
  to~$\2G$.  Again, the wheel in $G_1$ forms $m$ faces incident to
  $u$, and since~$G_1^\star$ is isomorphic to~$\2G$, the trees must be
  embedded such that each face contains exactly~$B$ bridges.  Since
  embedding~$T_i$ inside a face~$f$ puts~$a_i$ bridges into~$f$ and we
  have~$B/4 < a_i < B/2$ each face must contain exactly three trees.
  Thus the set of triplets determined by trees that are embedded in
  the same faces form a solution to \threepart.
  
  Clearly, the transformation can be computed in polynomial time, and
  thus \mpd is NP-hard.  Moreover, the graph $G_2$ can be made simple
  ($G_1$ is already simple) by replacing each bridge in $G_1$ and each
  loop in $G_2$ with a 4-wheel as depicted in
  Figure~\ref{fig:reduction}(d).  The resulting graphs $G_1'$ and
  $G_2'$ are obviously dual to each other if and only if $G_1$ and
  $G_2$ are dual to each other.  Moreover, $G_1'$ and $G_2'$ are
  simple, which concludes the proof.
\end{proof}

In the following we show how the above construction can be used to
show NP-completeness for {\sc Map} and {\sc Graph Self-Duality}.  To
this end, we use the \emph{adhesion} operation introduced by Servatius
and Christopher~\cite{sc-csdg-92}.  Let $v$ be a vertex of $G$
incident to a face $f$.  Then the adhesion of $G$ and its dual
$G^\star$ (with respect to $v$ and $f$) is obtained by identifying $v$
in $G$ and $f^\star$ in $G^\star$ with each other.  Servatius and
Christopher~\cite{sc-csdg-92} show that the adhesion of a plane graph
and its dual is graph self-dual.  Moreover, they implicitly show that
this adhesion is even map self-dual, although they do not mention it
explicitly.  To show the following theorem we essentially transform
the instance of \mpd consisting of the two graphs $G_1$ and $G_2$
described in the proof of Theorem~\ref{thm:mpd-np-compl} into an
equivalent instance of {\sc Map} and {\sc Graph Self-Duality} by
forming the adhesion of $G_1$ and $G_2$.

\begin{theorem}
  \label{thm:self-dual-NP-compl}
  {\sc Graph Self-Duality} and {\sc Map Self-Duality} are NP-complete.
\end{theorem}
\begin{proof}
  Clearly, {\sc Graph Self-Duality} ({\sc Map Self-Duality}) is in NP
  as we can guess an embedding for $G$ together with a bijection
  between the vertices of $G$ and the vertices of $G^\star$ and then
  check in polynomial time whether this bijection is an isomorphism
  (that preserves the embedding).

  Let $G_1$ and $G_2$ form an instance of \mpd obtained from an
  instance of \threepart as described in the proof of
  Theorem~\ref{thm:mpd-np-compl}.  Let $G$ be the graph obtained from
  $G_1$ and $G_2$ by identifying a vertex that is not the center of
  the wheel in $G_2$ with the vertex $u$ in $G_1$.  In the following
  we consider $G$ as an instance of {\sc Graph Self-Duality} and {\sc
    Map Self-Duality}.  We claim the following.
  \begin{compactdesc}
  \item[Claim 1] If $G$ is a {\sc Yes}-instance of {\sc Map
      Self-Duality}, it is a {\sc Yes}-instance of {\sc Graph
      Self-Duality}.
  \item[Claim 2] If $G_1$ and $G_2$ form a {\sc Yes}-instance of \mpd,
    then $G$ is a {\sc Yes}-instance of {\sc Map Self-Duality}.
  \item[Claim 3] If $G$ is a {\sc Yes}-instance of {\sc Graph
      Self-Duality}, then $G_1$ and $G_2$ form a {\sc Yes}-instance of
    \mpd.
  \end{compactdesc}
  % \begin{compactenum}
  % \item If $G$ is a {\sc Yes}-instance of {\sc Map Self-Duality}, then
  %   it is a {\sc Yes}-instance of {\sc Graph Self-Duality}.
  % \item If $G_1$ and $G_2$ form a {\sc Yes}-instance of \mpd, then $G$
  %   is a {\sc Yes}-instance of {\sc Map Self-Duality}.
  % \item If $G$ is a {\sc Yes}-instance of {\sc Graph Self-Duality},
  %   then $G_1$ and $G_2$ form a {\sc Yes}-instance of \mpd.
  % \end{compactenum}
  The three claims together show that the instance $G_1$ and $G_2$ of
  \mpd, the instance $G$ of {\sc Graph Self-Duality} and the instance
  $G$ of {\sc Map Self-Duality} are equivalent.

  Claim~1 is clear since being map self-dual is a stricter requirement
  then being graph self-dual.  For Claim~2 assume that $G_1$ and $G_2$
  form a {\sc Yes}-instance of \mpd, that is $G_1$ and $G_2$ admit
  embeddings such that they are dual to each other.  As the vertex $u$
  is incident to all faces in $G_1$ except for the face forming the
  center of the wheel in $G_1^\star$, it is in particular incident to
  the face dual to the vertex in $G_2$ chosen for the adhesion.  Thus
  it follows from the results by Servatius and
  Christopher~\cite{sc-csdg-92} that the adhesion $G$ of $G_1$ and
  $G_2$ is map self-dual.

  It remains to show Claim~3.  Let $G^\star$ be the dual graph of $G$
  with respect to a fixed embedding and let $\varphi \colon V(G)
  \longrightarrow V(G^\star)$ be a graph isomorphism between $G$ and
  $G^\star$.  As $G$ is the adhesion of $G_1$ and $G_2$ there is a
  unique vertex $v$ in $G$ belonging to $G_1$ and $G_2$, and a unique
  face $f$ incident to both graph $G_1$ and $G_2$.  Since $v$ was
  chosen to be $u$ in $G_1$, it is the only vertex in $G$ that is a
  cutvertex and the center of a wheel of size $m$.  Moreover, $f$ is
  the only cutvertex in $G^\star$ that can be the center of a wheel of
  size $m$.  Thus $\varphi$ has to map $v$ to $f$.  The blocks
  incident to $v$ are a block with degree~3 at $v$ stemming from
  $G_2$, $B$ loops also stemming from $G_2$, a block consisting of a
  wheel of size $m$ with center $v$ stemming from $G_1$ and $3m$
  attached trees stemming from $G_1$.  Similar the vertex $f$ in
  $G^\star$ is incident to a block having degree~3 at $f$ contained in
  $G_1^\star$, a set of loops stemming from the trees in $G_1$ (the
  number of loops depends on the embedding), a wheel of size $m$ with
  center $f$ contained in $G_2^\star$ and a set of bridges stemming
  from the loops at $G_2$.  Thus, $\varphi$ has to map all vertices in
  $G_1$ to vertices in $G_2^\star$ and all vertices in $G_2$ to
  vertices in $G_1^\star$.  This directly shows that $G_1$ and $G_2$
  form a {\sc Yes}-instance of \mpd, which concludes the proof.
\end{proof}

\section{Decomposition Trees and the SPQR-Tree}
\label{sec:spqr-tree}

A graph is \emph{connected} if there exists a path between any pair of
vertices.  A \emph{separating $k$-set} is a set of $k$ vertices whose
removal disconnects the graph.  Separating 1-sets and 2-sets are
\emph{cutvertices} and \emph{separation pairs}, respectively.  A
connected graph is \emph{biconnected} if it does not have a cut vertex
and \emph{triconnected} if it does not have a separation pair.  The
maximal biconnected components of a graph are called \emph{blocks}.

In the following we consider \emph{decomposition trees} of biconnected
planar graphs containing the \emph{SPQR-trees} introduced by Di
Battista and Tamassia~\cite{dt-omtc-96,dt-opt-96} as a special case.
Let $G$ be a planar biconnected graph and let $\{s, t\}$ be a
\emph{split pair}, that is either a separation pair or a pair of
adjacent vertices.  Let further $H_1$ and $H_2$ be two subgraphs of
$G$ such that $H_1 \cup H_2 = G$ and $H_1 \cap H_2 = \{s, t\}$.
Consider the tree $\mathcal T$ consisting of two nodes $\mu_1$ and
$\mu_2$ associated with the graphs $H_1 + (s, t)$ and $H_2 + (s, t)$,
respectively.  For each node $\mu_i$, the graph $H_i + (s, t)$
associated with it is the \emph{skeleton} of $\mu_i$, denoted by
$\skel(\mu_i)$, and the special directed edge $(s, t)$ is called
\emph{virtual edge}. The edge connecting the nodes $\mu_1$ and $\mu_2$
in $\mathcal T$ associates the virtual edge $\eps_1 = (s, t)$ in
$\skel(\mu_1)$ with the virtual edge $\eps_2 = (s, t)$ in
$\skel(\mu_2)$; we say that $\eps_1$ is the \emph{twin} of $\eps_2$
and vice versa.  Moreover, we say that $\eps_1$ in $\skel(\mu_1)$
\emph{corresponds} to the neighbor $\mu_2$ of $\mu_1$.  This can be
expressed as a bijective map $\corr_\mu \colon E(\skel(\mu))
\longrightarrow N(\mu)$ for each node $\mu$, where $E(\skel(\mu))$ and
$N(\mu)$ denote the set of edges in $\skel(\mu)$ and the set of
neighbors of $\mu$ in $\mathcal T$, respectively.  In the example
above we have $\corr(\eps_1) = \mu_2$ and $\corr(\eps_2) = \mu_1$ (the
index at $\corr$ is omitted as it is clear from the context).

The above described procedure is called \emph{decomposition} and can
of course be applied further to the skeletons of the nodes of
$\mathcal T$, leading to a larger tree with smaller skeletons.  The
decomposition can be undone by \emph{contracting} an edge in $\mathcal
T$.  Let $\{\mu, \mu'\}$ be an edge in $\mathcal T$ and let $\eps$ be
the virtual edge in $\skel(\mu)$ with $\corr(\eps) = \mu'$ having
$\eps'$ in $\skel(\mu')$ as twin.  The contraction of $\{\mu, \mu'\}$
replaces these two nodes by a single node with a skeleton obtained as
follows.  The skeletons $\skel(\mu)$ and $\skel(\mu')$ are
\emph{glued} together at the twins $\eps$ and $\eps'$ according to
their orientation, that is the source and target of $\eps$ is
identified with the source and target of $\eps'$, respectively.
Afterwards the resulting virtual edge is removed.  Applying the
contraction iteratively in an arbitrary order to edges in a
decomposition tree $\mathcal T$ yields a tree consisting of a single
node $\mu$ (which can be seen as trivial decomposition tree).  Then
the graph \emph{represented} by $\mathcal T$ is $\skel(\mu)$, which is
uniquely determined by~$\mathcal T$.

A \emph{reversed decomposition tree} is defined as a decomposition
tree with the only difference that in the decomposition step one of
the two twin edges is reversed and in the contraction step they are
glued together such that they point in opposite directions.  Note that
a reversed decomposition tree can be easily transformed into an
equivalent normal decomposition tree representing the same graph by
reversing one virtual edge for every pair of twin edges.

A special decomposition tree of a biconnected planar graph $G$ is the
SPQR-tree.  A decomposition tree is an SPQR-tree if each inner node is
either an S-, a P-, or an R-node whose skeletons contain only virtual
edges forming a cycle, a bunch of at least three parallel edges or a
triconnected planar graph, respectively, such that no two S-nodes and
no two P-nodes are adjacent.  Moreover, each leaf is a Q-node whose
skeleton consists of two vertices with one virtual and one normal edge
between them.  The \emph{reversed SPQR-tree} is defined analogous as a
special case of the reversed decomposition tree.  The SPQR-tree
$\mathcal T$ of a planar biconnected graph $G$ is unique up to the
reversal of pairs of virtual edges that are twins.  We can assume
without loss of generality that the virtual edges in the skeleton of
each P-node are oriented in the same direction and those in the
skeleton of each S-node form a directed cycle.

The SPQR-tree $\mathcal T$ of $G$ represents all planar embeddings of
$G$, as there is a bijection between these embeddings and the set of
all combinations of embeddings of the skeletons in $\mathcal T$.  Note
that the embedding choices for the skeletons consist of reordering the
parallel edges in a P-node and flipping the skeleton of an R-node.
The SPQR-tree of a biconnected planar graph can be computed in linear
time~\cite{gm-lti-00}.  Fixing the embeddings of skeletons in an arbitrary
decomposition tree $\mathcal T$ also determines a planar embedding of
the represented graph $G$.  However, there may be planar embeddings
that are not represented by $\mathcal T$.

\section{Succinct Representation of all Duals of a Biconnected Graph}
\label{sec:succ-repr-all-duals}

Let $G$ be a biconnected graph with SPQR-tree $\cal T$ and planar
embedding $\cal G$.  In the following we study the effects of changing
the embedding of $G$ on the corresponding dual graph~$G^\star$ of $G$.
To this end, we do not consider the graphs themselves but their
SPQR-trees.  More precisely, we first show how the SPQR-tree of the
dual graph $G^\star$ can be directly obtained from the SPQR-tree of
the primal graph $G$.  This can then be used to understand the effects
in the dual graph caused by changing the embedding of a skeleton
in~$\cal T$.

We first define the \emph{dual decomposition tree} ${\cal T}^\star$ of
the decomposition tree $\mathcal T$ representing $G$ (with respect to
a fixed embedding $\cal G$ of $G$ that can be represented by $\mathcal
T$).  It can then be shown that the dual decomposition tree represents
the dual graph $G^\star$ of $G$.  Essentially, ${\cal T}^\star$ is
obtained from $\cal T$ by replacing each skeleton with its directed
dual and interpreting the resulting tree as a reversed decomposition
tree.  More precisely, for each node $\mu$ in $\mathcal T$, the dual
decomposition tree ${\cal T}^\star$ contains a \emph{dual node}
$\mu^\star$ having the dual of $\skel(\mu)$ as skeleton.  An edge
$\eps^\star$ in $\skel(\mu^\star)$ dual to a virtual edge $\eps$ in
$\skel(\mu)$ is again virtual and oriented from right to left with
respect to the orientation of $\eps$.  Similarly, an edge dual to a
normal edge is also a normal edge in the dual skeleton.  Two virtual
edges in $\mathcal T^\star$ are twins if and only if their
corresponding primal edges are twins.  This implicitly has the effect
that $\corr(\eps)^\star = \corr(\eps^\star)$ holds.  Obviously, the
tree structures of $\mathcal T^\star$ and $\cal T$ are isomorphic with
respect to the map assigning each node in $\cal T$ to its dual node in
${\cal T}^\star$.  For the case in which $\mathcal T$ is the SPQR-tree
of $G$ we obtain the following. The dual of a triconnected skeleton is
triconnected, the dual of a (directed) circle is a bunch of parallel
edges (all directed in the same direction), and the dual of a normal
edge with a parallel virtual edge is a normal edge with a parallel
virtual edge.  Thus, if a node $\mu$ in $\cal T$ is an S-, P-, Q-, or
R-node, its dual node $\mu^\star$ in $\mathcal T^\star$ is a P-, S-,
Q-, or R-node, respectively.  This in particular shows that the dual
SPQR-tree is again an SPQR-tree and not just an arbitrary
decomposition tree.

\begin{lemma}
  \label{lem:dual-spqr-is-spqr-of-dual}
  Let $G$ be a biconnected planar graph with SPQR-tree $\cal T$ and
  planar embedding $\cal G$.  The dual SPQR-tree $\mathcal T^\star$
  with respect to $\mathcal G$ is the reversed SPQR-tree of the dual
  $G^\star$.

  % Thus, the SPQR-tree of the dual graph can essentially be obtained by
  % replacing each skeleton in the SPQR-tree of the primal graph by its
  % dual.

  % Let $G$ be a biconnected graph and let $\cal T$ be the SPQR-tree of $G$.
  % Further, let $\cal G$ be an embedding of $G$ and let $G^\star$ be the dual
  % graph of $G$ with respect to $\cal G$. Then, the SPQR-tree ${\cal
  %   T}^\star$ of $G^\star$ can be obtained from $\cal T$ by replacing each
  % skeleton with its dual and associating two dual virtual edges in
  % different skeletons with one another if the corresponding primal edges
  % are associated.
\end{lemma}
\begin{proof}
  We show a slightly more general result by dealing with arbitrary
  decomposition trees instead of SPQR-trees.  We show that first
  contracting an edge $\{\mu, \mu'\}$ in a decomposition tree $\cal T$
  into a node $\mu\mu'$ and then constructing the dual decomposition
  tree is equivalent to first constructing the decomposition tree
  $\mathcal T^\star$ and then contracting the edge $\{\mu^\star,
  \mu'^\star\}$ into $\mu^\star\mu'^\star$ (recall that ${\cal
    T}^\star$ is interpreted as reversed decomposition tree, thus the
  gluing operation contained in the contraction of $\{\mu^\star,
  \mu'^\star\}$ is reversed).  Applying this operation iteratively
  until the trees $\cal T$ and $\mathcal T^\star$ consist of single
  nodes then directly shows that the reversed decomposition tree
  $\mathcal T^\star$ represents the graph $G^\star$ dual to the graph
  $G$ represented by $\cal T$.

  Let $\eps$ and $\eps'$ be the virtual edges in $\skel(\mu)$ and
  $\skel(\mu')$ corresponding to the edge $\{\mu, \mu'\}$ in $\mathcal
  T$.  Let further $f_\ell$ and $f_r$, and $f_\ell'$ and $f_r'$ be the
  faces left and right to $\eps$ and $\eps'$ with respect to the
  orientation of $\eps$ and $\eps'$, respectively; see
  Figure~\ref{fig:contraction-in-dual}(a) and~(b).  We denote the
  graph $\skel(\mu) - \eps$ by $H$ and the graph $\skel(\mu^\star) -
  \eps^\star$ by $H^\star$ (note that $H^\star$ is not really the dual
  graph of $H$).  The graphs $H'$ and $H'^\star$ are defined similar.
  When contracting $\{\mu, \mu'\}$, first the virtual edges $\eps$ and
  $\eps'$ are glued together, that is $u$ and $v$ are identified with
  $u'$ and $v'$, respectively; see
  Figure~\ref{fig:contraction-in-dual}(c).  Obviously, the dual of the
  resulting graph can be obtained from $H^\star$ and $H'^\star$ by
  identifying $f_r$ with $f_\ell'$ and adding the edge $\{f_\ell,
  f_r'\}$ (or the other way round).  Finally, removing the edge $(u,
  v)$ contracts $f_\ell$ and $f_r'$ into a single vertex, see
  Figure~\ref{fig:contraction-in-dual}(d).  Thus the dual graph
  $\skel(\mu\mu')^\star$ of the resulting skeleton $\skel(\mu\mu')$
  can be obtained from $\skel(\mu)^\star$ and $\skel(\mu')^\star$ by
  removing the virtual edges $\eps^\star$ and $\eps'^\star$ and
  identifying their endpoints with each other, reversing their
  orientation.  As this is equal to contracting $\{\mu^\star,
  \mu'^\star\}$ in ${\cal T}^\star$, we have $\skel(\mu\mu')^\star =
  \skel(\mu^\star\mu'^\star)$, which concludes the proof.
\end{proof}

\begin{figure}
  \centering
  \includegraphics[page=1]{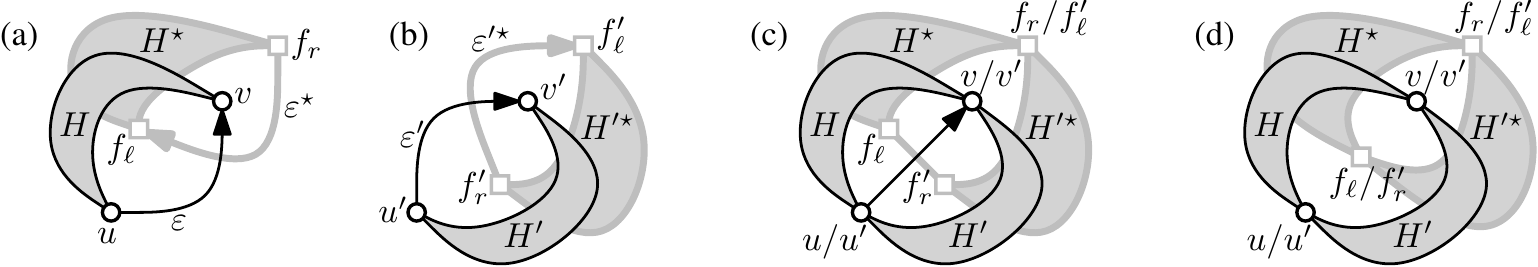}
  \caption{Glueing the virtual edge $\eps$ (a) and $\eps'$ (b)
    together (c) and removing the resulting edge~(d).}
  \label{fig:contraction-in-dual}
\end{figure}

The above results help investigating the effects a change in the
embedding of the graph $G$ has on its dual $G^\star$.  Flipping the
skeleton of an R-node and reordering the virtual edges in a P-node
give rise to the following two operations: \emph{reversal} of R-nodes and
\emph{restacking} of S-nodes.
A reversal applied on an R-node $\mu$ reverses the direction of all
the virtual edges in $\skel(\mu)$. As no other skeleton is changed by
this operation, this only affects how $\skel(\mu)$ is glued to the
skeletons of its adjacent nodes.
Let $\mu$ be an S-node with virtual edges $\eps_1,\dots,\eps_k$. A
restacking of $\mu$ picks an arbitrary ordering of
$\eps_1,\dots,\eps_k$ and glues their end-points such that they create
a directed cycle $C$ in that order. Then, the skeleton of $\mu$ is replaced
by $C$.

\begin{lemma}
  \label{lem:reversal-and-restacking}
  Let $\cal T$ be the SPQR-tree of a biconnected planar graph and let
  ${\cal T}^\star$ be the SPQR-tree of its dual, with respect to a
  fixed planar embedding.  Flipping an R-node and reordering a P-node
  in $\cal T$ corresponds to reversing its dual R-node and restacking
  its dual S-node, respectively.
\end{lemma}
\begin{proof}
  Due to Lemma~\ref{lem:dual-spqr-is-spqr-of-dual} we can work with
  the dual SPQR-tree instead of the SPQR-tree of the dual.  Obviously,
  flipping an R-node $\mu$ in $\cal T$ exchanges left and right in
  $\skel(\mu)$ and thus reverses the orientation of each virtual edge
  in $\skel(\mu^\star)$, where $\mu^\star$ is the node in ${\cal
    T}^\star$ dual to $\mu$.  Thus, flipping $\mu$ corresponds to a
  reversal of $\mu^\star$.  Similarly, reordering the virtual edges in
  the skeleton of a P-node $\mu$ has the effect that the virtual edges
  in its dual S-node $\mu^\star$ are restacked, yielding a different
  cycle.  Note that this cycle is again directed since the virtual
  edges in $\mu$ are still all oriented in the same direction.
\end{proof}

This shows that the SPQR-tree of the dual graph with respect to a
fixed embedding can be used to represent the dual graphs with respect
to all possible planar embeddings by allowing reversal and restacking
operations.  If we say that an SPQR-tree \emph{represents a set of
  dual graphs}, we implicitly allow reversal and restacking.  The
following theorem directly follows.

\begin{theorem}
  \label{thm:-dual-spqr-represents-duals}
  The dual SPQR-tree of a biconnected planar graph $G$ represents
  exactly the dual graphs of $G$.
  % Let $G$ be a biconnected planar graph with SPQR-tree $\mathcal T$.
  % Then its dual SPQR-tree $\mathcal T^\star$ represents all dual
  % graphs of $G$.
\end{theorem}

If we are only interested in the structure of the dual graph and not
in its embedding induced by the primal graph, we may also allow the
usual SPQR-tree operations, that is flipping R-nodes and reordering
the virtual edges in P-nodes.  In this case we can apply the reversal
operation not only to R-nodes but also to P-nodes (as this only
changes the embedding but not the graph).  Moreover, reversing a
Q-node does not change anything and the reversal of an S-node can be
seen as a special way of restacking it.  This observation can be used
to show the following lemma.

\begin{lemma}
  \label{lem:reversal-of-single-edge}
  Let $G$ be a biconnected planar graph and let $G^\star$ be its dual
  graph with SPQR-tree $\mathcal T^\star$ with respect to an embedding
  $\cal G$ of $G$.  Let $\mathcal T_\eps^\star$ be the SPQR-tree
  obtained from $\mathcal T^\star$ by reversing the orientation of the
  virtual edge $\eps$ in $\mathcal T^\star$ and let $G_\eps^\star$ be
  the graph represented by it.  Then there exists an embedding
  $\mathcal G_\eps$ of $G$ such that $G_\eps^\star$ is the dual graph
  of $G$ with respect to $\mathcal G_\eps$.
\end{lemma}
\begin{proof}
  Let $\mu$ be the node in $\mathcal T^\star$ containing the virtual
  edge $\eps$ and let $\corr(\eps) = \mu'$ be the neighbor of $\mu$
  corresponding to $\eps$.  Removing the edge $\{\mu, \mu'\}$ splits
  $\mathcal T^\star$ into two subtrees $\mathcal T_\mu^\star$ and
  $\mathcal T_{\mu'}^\star$.  We claim that the reversal of all nodes
  in one of these subtrees (no matter which one) yields an SPQR-tree
  $\mathcal T_{\mu\mu'}^\star$ representing $G_\eps^\star$.  Then it
  follows by Lemma~\ref{lem:reversal-and-restacking} and the
  observation above, that $G_\eps^\star$ is a dual graph of $G$.

  It remains to show the claim.  As it does not matter whether the
  orientation of $\eps$ or of its twin in $\mu'$ is changed, we can
  assume without loss of generality that all nodes in $\mathcal
  T_\mu^\star$ are reversed in $\mathcal T_{\mu\mu'}^\star$.  The
  graph represented by $\mathcal T_{\mu\mu'}^\star$ can be obtained by
  contracting the edges in an arbitrary order.  Contracting edges in
  the subtree $\mathcal T_{\mu'}^\star$ has the same effect as in the
  original graph, since $T_{\mu'}^\star$ was not changed.  Similarly,
  contracting an edge in $\mathcal T_\mu^\star$ also has the same
  effect as the orientation of both corresponding virtual edges is
  reversed.  Finally, when contracting the edge $\{\mu, \mu'\}$ the
  skeletons are glued together oppositely as $\eps$ is reversed
  whereas its twin remains the same.  Thus, reversing all nodes in
  $\mathcal T_\mu^\star$ is equivalent to reversing the orientation of
  $\eps$, which concludes the proof.
\end{proof}

Lemma~\ref{lem:reversal-and-restacking} and
Lemma~\ref{lem:reversal-of-single-edge} together yield the following
theorem.

\begin{theorem}
  \label{thm:dual-spqr-tree-rep}
  Two SPQR-trees represent the same set of dual graphs if and only if
  they can be transformed into each other by either using reversal and
  restacking operations, or by choosing an orientation of the virtual
  edges and restacking the skeletons of S-nodes.
\end{theorem}

\section{Equivalence Relation}
\label{sec:equivalence-relation}

We define the relation $\sim$ on the set of planar graphs as follows.
Two graphs $G_1$ and $G_2$ are related, i.e., $G_1 \sim G_2$, if and
only if $G_1$ and $G_2$ can be embedded such that they have the same
dual graph $G_1^\star = G_2^\star$.  We call $\sim$ the \emph{common
  dual relation}.

\begin{theorem}
  \label{thm:biconn-equiv-rel}
  The common dual relation $\sim$ is an equivalence relation on the
  set of biconnected planar graphs.  For a biconnected planar graph
  $G$, the set of dual graphs of $G$ is an equivalence class with
  respect to~$\sim$.
\end{theorem}
\begin{proof}
  Clearly, $\sim$ is symmetric and reflexive.  For the transitivity
  let $G_1$, $G_2$ and $G_3$ be three biconnected planar graphs such
  that $G_1 \sim G_2$ and $G_2 \sim G_3$.  Let further $\mathcal
  T_1^\star$, $\mathcal T_2^\star$ and $\mathcal T_3^\star$ be the
  dual SPQR-trees representing all duals of $G_1$, $G_2$ and $G_3$,
  respectively.  Due to $G_1 \sim G_2$ there exists a graph $G$ that
  is represented by $\mathcal T_1^\star$ and $\mathcal T_2^\star$.
  Since the SPQR-tree of a biconnected planar graph is unique (up to
  the reversal of virtual edges), it follows that $\mathcal T_1^\star$
  and $\mathcal T_2^\star$ are the same SPQR-trees representing the
  same sets of duals.  The same argument shows that $G_2$ and $G_3$
  have the same set of dual graphs, due to $G_2 \sim G_3$.  Thus, also
  $G_1$ and $G_3$ have exactly the same set of dual graphs, which
  yields $G_1 \sim G_3$.

  For the second statement, let $C^\star$ be the set of dual graphs of
  $G$.  Clearly, for $G_1^\star, G_2^\star \in C^\star$ the graph $G$
  is a common dual, thus $G_1^\star \sim G_2^\star$.  On the other
  hand, let $G_1^\star \in C^\star$ and $G_1^\star \sim G_2^\star$.
  By the above argument, $G_1^\star$ and $G_2^\star$ have the same set
  of dual graphs.  Thus $G$ is a dual graph of $G_2^\star$ yielding
  $G_2^\star \in C^\star$.
\end{proof}

Theorem~\ref{thm:biconn-equiv-rel} shows that the equivalence class
$C$ of a biconnected planar graph $G$ with respect to the common dual
relation is exactly the set of dual graphs that is represented by the
SPQR-tree $\mathcal T$ of $G$.  The dual SPQR-tree $\mathcal T^\star$
of $G$ also represents a set of dual graphs forming the equivalence
class $C^\star$.  We say that $C^\star$ is the \emph{dual equivalence
  class} of $C$.  Given an arbitrary graph $G \in C$ and an arbitrary
graph $G^\star \in C^\star$ then $G$ and $G^\star$ can be embedded
such that they are dual to each other since $C^\star$ contains exactly
the graphs that are dual to $G$.  The problems \mpd and {\sc Graph
  Self-Duality} can be reformulated in terms of the equivalence
classes of the common dual relation.  Two biconnected planar graphs
are a {\sc Yes}-instance of \mpd if and only if their equivalence
classes are dual to each other.  A biconnected planar graph is graph
self-dual if and only if its equivalence class is dual to itself.
This in particular means that either each or no graph in an
equivalence class is graph self-dual.

Although it might seem quite natural that the common dual relation is
an equivalence relation, this is not true for general planar graphs.
This fact is stated in the following theorem.

\begin{theorem}
  \label{thm:no-equivalence-relation}
  The common dual relation $\sim$ is not transitive on the set of
  planar graphs.
\end{theorem}
\begin{proof}
  Consider the graph $G_1$ consisting of a triconnected planar graph
  with an additional loop as depicted in
  Figure~\ref{fig:no-equivalence-relation}(a).  Its dual graph is a
  triconnected component with a bridge attached to it.  In the graph
  $G_2$ the loop is attached to a different vertex, see
  Figure~\ref{fig:no-equivalence-relation}(b).  However, in both
  graphs $G_1$ and $G_2$, the loop can be embedded into the same face
  of the triconnected component, yielding the same dual graph (with a
  different embedding).  Thus, $G_1$ and $G_2$ have a common dual,
  i.e., $G_1 \sim G_2$ holds.  The same argument yields that $G_2$
  (with respect to the embedding in
  Figure~\ref{fig:no-equivalence-relation}(c)) and $G_3$ have a common
  dual graph, i.e., $G_2 \sim G_3$.  However, $G_1$ and $G_3$ do not
  have a common dual for the following reason.  Let $v_1$ and $v_3$ be
  the vertices in $G_1$ and $G_3$ incident to the loop, respectively.
  The only embedding choice in $G_1$ and $G_3$ is to embed the loop
  into one of the faces incident to $v_1$ and $v_3$, respectively.  In
  the dual graphs this has the effect that the bridge is attached to
  the corresponding vertex.  Since all faces incident to $v_1$ have
  degree~3 and all faces incident to $v_3$ have degree~4 or~5, the
  resulting dual graphs cannot be isomorphic.  Thus, $G_1 \not\sim
  G_3$ even though $G_1 \sim G_2 \sim G_3$ holds.
\end{proof}

\begin{figure}
  \centering
  \includegraphics{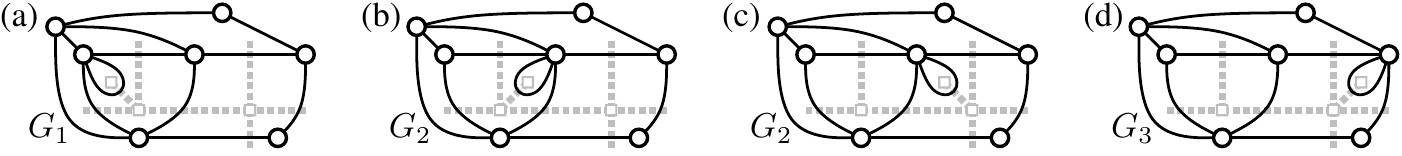}
  \caption{Illustration of Theorem~\ref{thm:no-equivalence-relation}}
  \label{fig:no-equivalence-relation}
\end{figure}

\section{Solving \mpd for Biconnected Graphs}
\label{sec:solv-mpd-biconn}

Due to Theorem~\ref{thm:-dual-spqr-represents-duals} the problem \mpd
can be rephrased as follows.

\begin{corollary}
  \label{cor:mpd-reduces-to-same-duals-of-spqr}
  Let $G_1$ and $G_2$ be two biconnected planar graphs with SPQR-trees
  $\mathcal T_1$ and $\mathcal T_2$, respectively.  There is an
  embedding $\mathcal G_1$ of $G_1$ such that $G_2$ is dual to $G_1$
  with respect to $\mathcal G_1$ if and only if $\mathcal T_2$ and the
  dual SPQR-tree $\mathcal T_1^\star$ represent the same set of dual
  graphs.
\end{corollary}

In the following we define what it means for two SPQR-trees to be dual
isomorphic and show that they are isomorphic in that sense if and only
if they represent the same set of dual graphs.  Afterwards, we show
that testing the existence of such an isomorphism reduces to testing
graph isomorphism for planar graphs, which then solves \mpd.
Figure~\ref{fig:overview-comm-diag}(a) sketches this strategy.

\begin{figure}
  \centering
  \includegraphics[page=2]{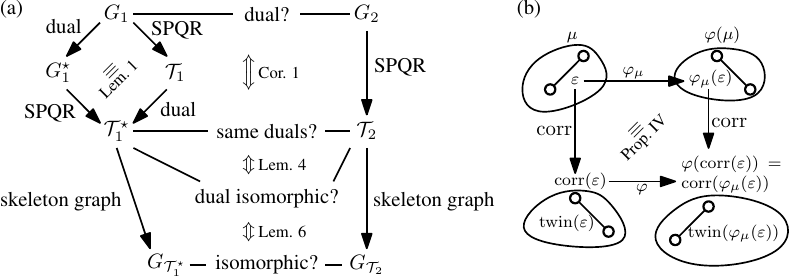}
  \caption{(a) Overview of our strategy.  (b)
    Commutative diagram illustrating Property~\ref{itm:prop-4}.}
  \label{fig:overview-comm-diag}
\end{figure}

For two graphs $G$ and $G'$ with vertices $V(G)$ and $V(G')$ and edges
$E(G)$ and $E(G')$, respectively, a map $\varphi \colon V(G)
\longrightarrow V(G')$ is a \emph{graph isomorphism} if it is
bijective and $\{u, v\} \in E(G)$ if and only if $\{\varphi(u),
\varphi(v)\} \in E(G')$.  A graph isomorphism $\varphi$ induces a
bijection between $E(G)$ and $E(G')$ and we use $\varphi(e)$ for $e
\in E(G)$ to express this bijection.  Note that we consider the edges
to be undirected, thus fixing $\varphi(\cdot)$ only for the edges does
not determine a map for the vertices.  As the SPQR-tree has more
structure than a normal tree, we require some additional properties.
A \emph{dual SPQR-tree isomorphism} between two SPQR-trees $\mathcal
T$ and $\mathcal T'$ consists of several maps.
First, a map $\varphi \colon V(\mathcal T) \longrightarrow V(\mathcal T')$
such that
\begin{compactenum}[(I)]
\item \label{itm:prop-1}$\varphi$ is a graph isomorphism between
  $\mathcal T$ and $\mathcal T'$; and
\item \label{itm:prop-2}for each node $\mu \in V(\mathcal T)$, the node
  $\varphi(\mu) \in V(\mathcal T')$ is of the same type.
  \newcounter{isomorphism-properties}
  \setcounter{isomorphism-properties}{\value{enumi}}
\end{compactenum}
Second, a map $\varphi_\mu \colon V(\skel(\mu)) \longrightarrow
V(\skel(\varphi(\mu)))$ for every R-node $\mu$ in $\mathcal T$ such
that
\begin{compactenum}[(I)]
  \setcounter{enumi}{\value{isomorphism-properties}}
\item \label{itm:prop-3}$\varphi_\mu$ is a graph isomorphism between
  $\skel(\mu)$ and $\skel(\varphi(\mu))$; and
\item \label{itm:prop-4} $\corr(\varphi_\mu(\eps)) =
  \varphi(\corr(\eps))$ holds for every virtual edge $\eps$ in
  $\skel(\mu)$.
\end{compactenum}
If there is a dual SPQR-tree isomorphism between $\mathcal T$ and
$\mathcal T'$, then we say that $\mathcal T$ and $\mathcal T'$ are
\emph{dual isomorphic}.  Note that Property~\ref{itm:prop-4} is a
quite natural requirement and one would usually require it also for
S-nodes (for P-nodes it does not make sense since every permutation is
an isomorphism on its skeleton); see
Figure~\ref{fig:overview-comm-diag}(b) for a commutative diagram
illustrating Property~\ref{itm:prop-4}.  However, not requiring it for
S-nodes has the effect that restacking their skeletons is implicitly
allowed.  As the graph isomorphisms $\varphi_\mu(\cdot)$ do not care
about the orientation of virtual edges it is also implicitly allowed
to reverse them.  These observations lead to the following lemma
showing that this definition of dual SPQR-tree isomorphism is well
suited for our purpose.

\begin{lemma}
  \label{lem:same-duals-iff-dual-isomorphic}
  Two SPQR-trees represent the same set of dual graphs if and only if
  they are dual isomorphic.
\end{lemma}
\begin{proof}
  Let $\mathcal T$ and $\mathcal T'$ be two SPQR-trees representing
  the same set of dual graphs.  By
  Theorem~\ref{thm:dual-spqr-tree-rep} this implies that they can be
  transformed into each other using reversal and restacking
  operations.  Clearly, the identity map, mapping $\mathcal T$ and
  each of its skeletons to itself, is a dual SPQR-tree isomorphism.
  It remains a dual SPQR-tree isomorphism when restacking an S-node,
  since Properties~\ref{itm:prop-1},~\ref{itm:prop-2} are independent
  from the skeletons and Properties~\ref{itm:prop-3},~\ref{itm:prop-4}
  are only required for R-nodes.  Moreover, the reversal of an R-node
  preserves Properties~\ref{itm:prop-1}--\ref{itm:prop-4} since our
  definition of graph isomorphism considers edges to be undirected.
  It follows that $\mathcal T$ and $\mathcal T'$ are dual isomorphic.

  For the opposite direction assume that $\varphi$ together with
  $\varphi_{\mu_1}, \dots, \varphi_{\mu_k}$ is a dual SPQR-tree
  isomorphism from $\mathcal T$ to $\mathcal T'$.  For every virtual
  edge $\eps$ in an R-node $\mu$ the map $\varphi_\mu$ determines
  whether the orientation of $\eps$ has to be reversed to match the
  orientation of $\varphi_\mu(\eps)$.  Moreover, how $\varphi$ maps
  the neighbors of an S-node $\mu$ to the neighbors of $\varphi(\mu)$
  determines a restacking operation transforming $\skel(\mu)$ into
  $\skel(\varphi(\mu))$.  It follows that $\mathcal T$ can be
  transformed into $\mathcal T'$ by applying restacking and reversal
  operations.  Hence, $\mathcal T$ and $\mathcal T'$ represent the
  same set of dual graphs, which concludes the proof.
\end{proof}

In the following we show how the question of whether two SPQR-trees
are dual isomorphic can be reduced to the graph isomorphism problem on
planar graphs, which can be solved in linear time~\cite{hw-ltaipg-74}.
We define the \emph{skeleton graph} $G_{\mathcal T}$ of an SPQR-tree
$\mathcal T$ as follows.  For each node $\mu$ in $\mathcal T$ there is
a subgraph $H_\mu$ in $G_{\mathcal T}$ and for each edge $\{\mu,
\mu'\}$ in $\mathcal T$ the skeleton graph contains an edge connecting
$H_\mu$ and $H_{\mu'}$.  In the following we describe the subgraphs
$H_\mu$ for the cases that $\mu$ is an S-, a P-, a Q- or an R-node and
define \emph{attachment vertices} that are incident to the edges
connecting $H_\mu$ to other subgraphs.

If $\mu$ is an S- or a P-node, the subgraph $H_\mu$ contains only one
attachment vertex $v_\mu$ and all subgraphs stemming from neighbors of
$\mu$ are attached to $v_\mu$.  To distinguish between S- and P-nodes,
small non-isomorphic subgraphs called \emph{tags} are additionally
attached to $v_\mu$, see Figure~\ref{fig:skeleton-graph}(s) and~(p),
respectively.  If $\mu$ is a Q-node, then $H_\mu$ is a single
attachment vertex, see Figure~\ref{fig:skeleton-graph}(q).  Note that
$\mu$ is a leaf in $\mathcal T$ and thus $H_\mu$ is also a leaf in
$G_{\mathcal T}$.  If $\mu$ is an R-node, $H_\mu$ is the skeleton
$\skel(\mu)$, where additionally every virtual edge $\eps$ is
subdivided by an attachment vertex $v_\eps$, see
Figure~\ref{fig:skeleton-graph}(r) for the case in which $\skel(\mu)$ is
$K_4$.  The subgraph $H_{\corr(\eps)}$ stemming from the neighbor
$\corr(\eps)$ of $\mu$ is attached to $H_\mu$ over the attachment
vertex $v_\eps$.

\begin{lemma}
  \label{lem:skeleton-graph-planar}
  The skeleton graph of an SPQR-tree is planar and can be computed in
  linear time.
\end{lemma}
\begin{proof}
  Clearly, the skeleton graph $G_{\mathcal T}$ of an SPQR-tree
  $\mathcal T$ can be computed in linear time by processing each node
  $\mu$ separately to compute the subgraph $H_\mu$ consuming time
  linear in the size of $\skel(\mu)$.  Note that this implicitly shows
  that the size of $G_{\mathcal T}$ is linear.

  Let $\mathcal T$ be an SPQR-tree rooted at an arbitrary node.  The
  skeleton graph $G_{\mathcal T}$ can be embedded in a planar way by
  embedding the subgraphs corresponding to the nodes in $\mathcal T$
  top-down with respect to the chosen root.  Obviously, every subgraph
  in $G_{\mathcal T}$ corresponding to a node in $\mathcal T$ is
  planar, thus we can start by embedding the subgraph corresponding to
  the root arbitrarily.  Let $\mu$ be a non-root node in $\mathcal T$
  and let $\mu'$ be its parent.  If $\mu$ is not an R-node, $H_\mu$
  can be embedded with its only attachment vertex on the outer face.
  If $\mu$ is an R-node, $H_\mu$ can be embedded with the attachment
  vertex corresponding to the parent $\mu'$ of $\mu$ in $\mathcal T$
  on the outer face.  Thus, in any case, $H_\mu$ can be placed inside
  a face incident to the attachment vertex stemming from $\mu'$
  corresponding to $\mu$, yielding a planar drawing.
\end{proof}

\begin{figure}
  \centering
  \includegraphics{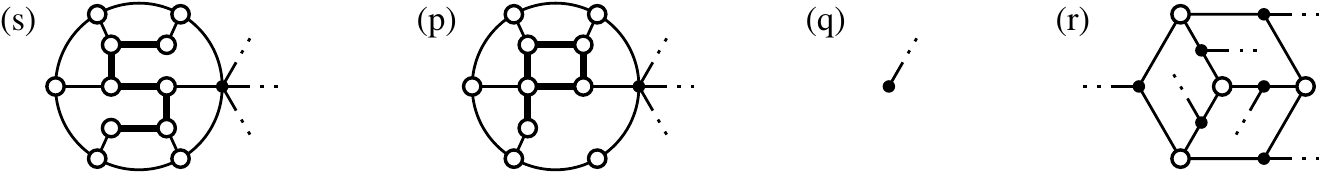}
  \caption{The subgraphs $H_\mu$ of the skeleton graph depending on
    the type of the node $\mu$.  The small black vertices are the
    attachment vertices.}
  \label{fig:skeleton-graph}
\end{figure}

\begin{lemma}
  \label{lem:dual-isomorphic-iff-skel-graphs-iso}
  Two SPQR-trees are dual isomorphic if and only if their skeleton
  graphs are isomorphic.
\end{lemma}
\begin{proof}
  Let $\mathcal T$ and $\mathcal T'$ be two SPQR-trees and let
  $\varphi$ together with $\varphi_{\mu_1}, \dots, \varphi_{\mu_k}$ be
  a dual SPQR-tree isomorphism between them.  We show how this induces
  a graph isomorphism $\varphi_G$ between the skeleton graphs
  $G_{\mathcal T}$ and $G_{\mathcal T'}$.  If $\mu$ is an S-, P- or
  Q-node, then its corresponding subgraph in $H_\mu$ only contains a
  single attachment vertex $v_\mu$.  Since $\varphi(\mu)$ is of the
  same type (due to Property~\ref{itm:prop-2} of dual SPQR-tree
  isomorphisms), the subgraph $H_{\varphi(\mu)}$ also contains a
  single attachment vertex $v_{\varphi(\mu)}$ and we set
  $\varphi_G(v_\mu) = v_{\varphi(\mu)}$.  For S- and P-nodes we
  additionally simply map their tags isomorphically to one another.
  For the case that $\mu$ is an R-node, the map $\varphi_\mu$ is a
  graph isomorphism between $\skel(\mu)$ and $\skel(\varphi(\mu))$
  (Property~\ref{itm:prop-3}).  Thus, it induces a graph isomorphism
  between $H_\mu$ and $H_{\varphi(\mu)}$ since these subgraphs are
  obtained from $\skel(\mu)$ and $\skel(\varphi(\mu))$, respectively,
  by subdividing each virtual edge.  It remains to show that
  $\varphi_G$ respects the edges between attachment vertices of
  different subgraphs.  Since $\varphi$ is a graph isomorphism
  (Property~\ref{itm:prop-1}), attachment vertices of two subgraphs of
  $G_{\mathcal T}$ are connected if and only if the corresponding
  subgraphs in $G_{\mathcal T'}$ are connected.  Moreover,
  Property~\ref{itm:prop-4} ensures that for a subgraph stemming from
  an R-node the right attachment vertices are chosen (for other nodes
  this is clear since their subgraphs have unique attachment
  vertices).
  
  For the opposite direction, assume $\varphi_G$ is a graph
  isomorphism between $G_{\mathcal T}$ and $G_{\mathcal T'}$.  Let
  $H_\mu$ be the subgraph stemming from a node $\mu$ in $\mathcal
  T$.  As $H_\mu$ is a block in $G_{\mathcal T}$ (or a leaf if $\mu$
  is a Q-node), it has to be mapped to a block in $G_{\mathcal T'}$.
  As all edges in $G_{\mathcal T'}$ connecting attachment vertices of
  subgraphs stemming from different nodes are bridges, all vertices in
  $H_\mu$ have to be mapped to vertices in $H_{\mu'}$ for some node
  $\mu'$ in $\mathcal T'$.  This defines the map $\varphi$ by setting
  $\varphi(\mu) = \mu'$.  Clearly, $\varphi$ is a graph isomorphism
  between $\mathcal T$ and $\mathcal T'$, since two subgraphs in a
  skeleton graph are connected by an edge if and only if the
  corresponding nodes in its SPQR-tree are adjacent, thus, $\varphi$
  satisfies Property~\ref{itm:prop-1}.  Since the only leaves in a
  skeleton graph stem from Q-nodes, $\varphi(\mu)$ is a Q-node if and
  only if $\mu$ is a Q-node.  Let $v$ be an attachment vertex stemming
  from an inner node $\mu$ in $\mathcal T$.  Then $v$ is a cutvertex
  and, since every cutvertex in a skeleton graph is an attachment
  vertex, $\varphi_G(v)$ is also an attachment vertex in $G_{\mathcal
    T'}$.  The vertex $v$ has degree~3 if and only if $\mu$ is an
  R-node, thus $\varphi$ maps R-nodes to R-nodes.  Moreover, if $\mu$
  is an S-node, $v$ cannot be mapped to an attachment vertex stemming
  from a P-node, since the tags attached to S- and P-nodes are not
  isomorphic.  Hence, $\varphi$ maps S- and P-nodes to S- and P-nodes,
  respectively, and thus satisfies Property~\ref{itm:prop-2}.
  
  To obtain a dual SPQR-tree isomorphism, it remains to define a map
  $\varphi_\mu$ for each R-node $\mu$ in $\mathcal T$ that satisfies
  Properties~\ref{itm:prop-3} and~\ref{itm:prop-4}.  As observed
  before, $\varphi_G$ defines a bijection between the vertices in the
  subgraph $H_\mu$ stemming from $\mu$ and the vertices in
  $H_{\varphi(\mu)}$ stemming from $\varphi(\mu)$.  As $H_\mu$ and
  $H_{\varphi(\mu)}$ are the skeletons $\skel(\mu)$ and
  $\skel(\varphi(\mu))$ (with a subdivision vertex on each virtual
  edge), $\varphi_G$ defines a graph isomorphism $\varphi_\mu$ between
  $\skel(\mu)$ and $\skel(\varphi(\mu))$ (satisfying
  Property~\ref{itm:prop-3}).  To show that Property~\ref{itm:prop-4}
  holds consider a virtual edge $\eps$ in $\skel(\mu)$ and the
  attachment vertex $v_\eps$ in $H_\mu$ stemming from it.  Let further
  denote $v_{\corr(\eps)}$ the attachment vertex in $H_{\corr(\eps)}$
  such that $G_{\mathcal T}$ contains the edge $\{v_\eps,
  v_{\corr(\eps)}\}$.  Then $\varphi_G$ maps $\{v_\eps,
  v_{\corr(\eps)}\}$ to an edge $\{\varphi_G(v_\eps),
  \varphi_G(v_{\corr(\eps)})\}$ in $G_{\mathcal T'}$.  Since
  $\varphi_G(v_\eps) = v_{\varphi_\mu(\eps)}$ holds by the definition
  of $\varphi_\mu$, and $\varphi_G(v_{\corr(\eps)})$ stems from the
  node $\varphi(\corr(\eps))$ by the definition of $\varphi$, we have
  that $\corr(\varphi_\mu(\eps)) = \varphi(\corr(\eps))$ by the
  definition of the skeleton graph $G_{\mathcal T'}$.  As this
  establishes Property~\ref{itm:prop-4}, it concludes the proof.
\end{proof}

\begin{figure}[tb]
  \centering
  \includegraphics[page=1]{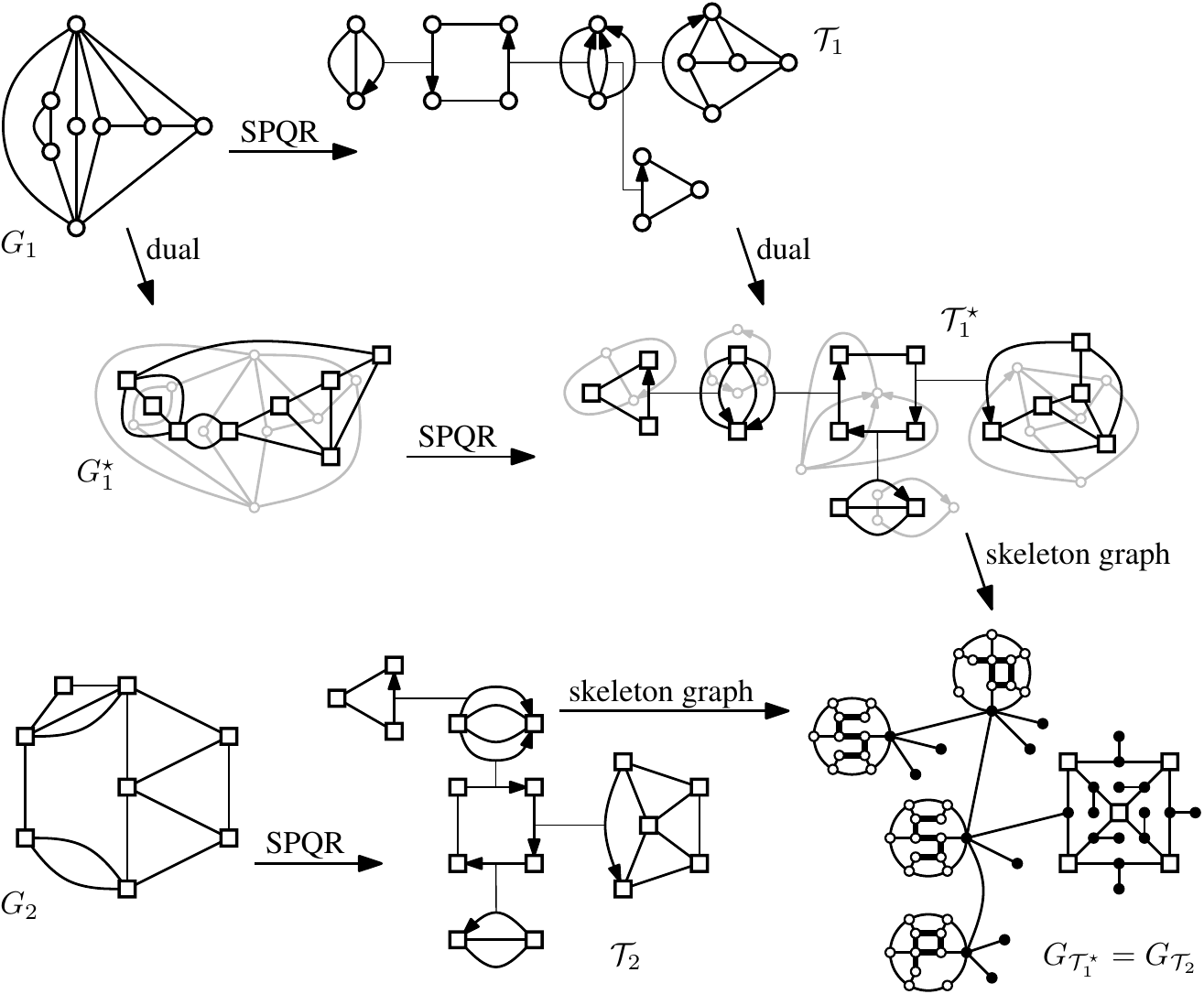}
  \caption{First building the dual graph $G_1^\star$ of $G_1$ (with
    respect to a fixed embedding) and then building its SPQR-tree or
    first building its SPQR-tree $\mathcal T_1$ and then its dual
    SPQR-tree yields the same tree $\mathcal T_1^\star$
    (Lemma~\ref{lem:dual-spqr-is-spqr-of-dual}).  The graphs $G_1$ and
    $G_2$ are dual to each other (with respect to at least one pair of
    embeddings) if and only if $\mathcal T_1^\star$ and $\mathcal T_2$
    represent the same set of duals
    (Corollary~\ref{cor:mpd-reduces-to-same-duals-of-spqr}), which is
    the case if and only if their skeleton graphs $G_{\mathcal
      T_1^\star}$ and $G_{\mathcal T_2}$ are isomorphic
    (Lemma~\ref{lem:same-duals-iff-dual-isomorphic} and
    Lemma~\ref{lem:dual-isomorphic-iff-skel-graphs-iso}).}
  \label{fig:overview}
\end{figure}

\begin{theorem}
  \mpd can be solved in linear time for biconnected planar graphs.

  % There is a linear-time algorithm for \mpd when the input graphs are
  % biconnected.
\end{theorem}
\begin{proof}
  See Figure~\ref{fig:overview} for an example.
  Corollary~\ref{cor:mpd-reduces-to-same-duals-of-spqr} states that
  \mpd can be solved by testing whether two SPQR-trees (that can be
  computed in time linear in the size of the input
  graphs~\cite{gm-lti-00}) represent the same set of dual graphs.  By
  Lemma~\ref{lem:same-duals-iff-dual-isomorphic} it is equivalent to
  test whether these two SPQR-trees are dual isomorphic, which can be
  done by testing whether their skeleton graphs are isomorphic, due to
  Lemma~\ref{lem:dual-isomorphic-iff-skel-graphs-iso}.  The skeleton
  graph of an SPQR-tree is planar and has linear size, see
  Lemma~\ref{lem:skeleton-graph-planar}.  Hence, we can use the linear
  time algorithm for testing whether two planar graphs are isomorphic
  by Hopcroft and Wong~\cite{hw-ltaipg-74} yielding a linear time
  algorithm solving \mpd.
\end{proof}

% As {\sc Graph Self-Duality} is a special case of \mpd, we obtain the
% following corollary.

\begin{corollary}
  \label{cor:self-dual-biconn}
  {\sc Graph Self-Duality} can be solved in linear time for
  biconnected planar graphs.

  % There is a linear-time algorithm for {\sc Graph Self-Duality} when
  % the input graph is biconnected.
\end{corollary}

\section{Conclusion}
\label{sec:conclusion}

In this paper we defined and studied the problem \mpd of testing whether,
given two graphs $\1G$ and $\2G$, there exists an embedding of $\1G$ such
that the corresponding dual graph is isomorphic to $\2G$. We proved that
\mpd is NP-complete in the general case, while it is solvable
in polynomial (actually linear) time for biconnected planar graphs.

The interest on this problem is twofold. On one hand, it represents a
new step in the fundamental theory of planar graphs isomorphism, also
testified by the fact that, as a side effect, it provides the same
results for the well-known problem of testing {\sc Graph
  Self-Duality}~\cite{sc-csdg-92,ar-ccsdsp-92}. On the other hand, it
could be seen as a single example among a pletora of problems that
require to find a dual graph of $\1G$ satisfying certain
properties. In this direction, we believe that the definition of the
new data-structure \emph{dual SPQR-tree} and of the operations that
can be applied on it to efficiently handle all the duals of a
biconnected planar graph could be considered as a main result of this
paper, independently of its application to solve \mpd, since it could
potentially be used to tackle many other problems of the same type.

As remarked above, the results we obtained on \mpd can be extended to
{\sc Graph Self-Duality}, asking whether a given graph $G$ can be
embedded in such a way that the corresponding dual is isomorphic to $G$. 
The restricted version {\sc Map Self-Duality}~\cite{ss-sdg-96}
of {\sc Graph Self-Duality} requires the embedding of $G$ to be preserved
in the isomorphism with the corresponding dual. We could prove that the
NP-completeness result for \mpd extends to {\sc Map Self-Duality}, but we
could not prove the same for the polynomial-time testing algorithm. Hence,
we leave as an open problem the question whether {\sc Map Self-Duality} can
be solved efficiently for biconnected planar graphs.

\clearpage
\newpage

\bibliographystyle{plain}
\bibliography{dual_embedding}

\end{document}